\documentclass[10pt,journal]{IEEEtran}
\usepackage[font=small,labelsep=space]{caption}
\usepackage{verbatim}
\usepackage{graphicx}
\usepackage{subfig}
\usepackage{multirow}
\usepackage{amssymb}
\usepackage{amsmath}
\usepackage{amsthm}
\usepackage{amsfonts}
\usepackage{paralist}
\usepackage{url}
\usepackage{epstopdf}
\usepackage{float}
\usepackage{color}
\usepackage{cite}
\usepackage{paralist}
\usepackage[noend]{algpseudocode}
\usepackage{algorithmicx}
\usepackage{algorithm}
\usepackage{mathrsfs}
\usepackage{amsmath}
\usepackage{hyperref}
\hypersetup{hypertex=true,
            linkcolor=red,
            anchorcolor=red,
            citecolor=red}
\let\itemize\compactitem
\let\enditemize\endcompactitem
\let\enumerate\compactenum
\let\endenumerate\endcompactenum

\begin{document}
\title{UAV Placement for Real-time Video Acquisition: A Tradeoff between Resolution and Delay}
\author{Xiao-Wei~Tang, \emph{Member, IEEE}, and Xin-Lin Huang, \emph{Senior Member, IEEE}
}
\maketitle
\begin{abstract}
Recently, UAVs endowed with high mobility, low cost, and remote control have promoted the development of UAV-assisted real-time video/image acquisition applications, which have a high demand for both transmission rate and image resolution. However, in conventional vertical photography model, the UAV should fly to the top of ground targets (GTs) to capture images, thus enlarge the transmission delay. In this paper, we propose an oblique photography model, which allows the UAV to capture images of GTs from a far distance while still satisfying the predetermined resolution requirement. Based on the proposed oblique photography model, we further study the UAV placement problem in the cellular-connected UAV-assisted image acquisition system, which aims at minimizing the data transmission delay under the condition of satisfying the predetermined image resolution requirement. Firstly, the proposed scheme is first formulated as an intractable non-convex optimization problem. Then, the original problem is simplified to obtain a tractable suboptimal solution with the help of the block coordinate descent and the successive convex approximation techniques. Finally, the numerical results are presented to show the effectiveness of the proposed scheme. The numerical results have shown that the proposed scheme can largely save the transmission time as compared to the conventional vertical photography model. 

\end{abstract}

\vspace{0.2in}
\section{Introduction}
The great progress of aviation, energy and artificial intelligence (AI) technology promotes the rapid development of unmanned aerial vehicles (UAVs), empowering them many advantages including low cost, controllable mobility, and line-of-sight (LoS) link with ground users \cite{R1}. At the same time, UAVs are endowed with the ability of real-time ultra high-definition image transmission in the fifth generation (5G) of mobile network, thus giving birth to many UAV-assisted image/video acquisition applications such as live broadcast, disaster monitoring, agriculture precision, and virtual reality/actual reality (VR/AR) \cite{Tang1}. 

Different from conventional UAV-assisted applications, e.g., remote sensing, the emerging UAV-assisted image/video acquisition applications need to transmit the captured image/video data back to the base station (BS) in real time, which are thus facing severe challenges. Firstly, these applications have an extremely large demand on bandwidth due to the huge amount of video/image amount \cite{Tang4}. Then, in general, these applications usually bear rigorous quality of experience (QoE) requirements since users expect to receive videos with low frame loss rate, tolerable end-to-end delay and little jitter \cite{Tang3}. Last but not least, the performance of these applications is still restricted by UAV's communication range and flight endurance owing to the limited power supply \cite{You1}. 

The challenges mentioned above are becoming irreconcilable when adopting conventional  vertical photography technique to capture images/videos. To be specific, in conventional vertical photography model, the UAV can only capture images at the top of the ground target to make sure it is located in the center of the image/video so that users can easily focus on it \cite{Sai}. Although this can provide high-definition images for users, the transmission rate will be very low when the GT is far away from the BS. Even worse, once the distance between the GT and the BS is greater than a certain threshold, the communication between the UAV and the BS will be interrupted due to the limited transmit power of the UAV, which seriously degrading user's QoE. Fortunately, oblique photography makes up for the limitation that images can only be photographed from vertical angles in the past \cite{He}. Specifically, the UAV doesn't need to fly to the point above the GT, but can choose a location between the BS and the GT when capturing images/videos, thus providing a flexible tradeoff between image resolution and communication quality.

A handful of research work on the oblique photography has been done in recent years \cite{Ma, Dai, Hohle, Zhou, Aghaei, Lin, Zhang}. Specifically, the concept of the oblique photography first appeared in aerial survey for visualization purposes. By carrying multiple sensors on the UAV and collecting images from 5 different angles including 1 vertical and 4 oblique angles at the same time, the real and intuitive image effect that conforms to human vision can be generated via a series of data processing methods such as multi-vision image joint adjustment and dense matching of multi-vision images \cite{Ma, Dai}. In addition to aerial survey, object measurement is also one of the main uses of obilque photography. H\"ohle et al. proposed to measuring the distances, coordinates, elevations or areas of objects via oblique images \cite{Hohle}. Zhou et al. acquired images of plantations with different ages via UAV oblique photography and then extracted tree heights according to reconstructed three dimension (3D) point clouds \cite{Zhou}. Aghaei et al. employed the UAV to fly over a test laboratory in order to capture images at different altitudes, in order to investigate the correlation between aerial image capture altitude and potential defect identification on photovoltaic modules \cite{Aghaei}. Lin et al. adopted an electric fixed-wing UAV loaded with a digital camera to take oblique photographs of a sparse subalpine coniferous forest in the source region, aiming at extracting individual tree heights with the help of generated point cloud data obtained from the overlapping photographs \cite{Lin}. Zhang et al. proposed a UAV-based panoramic oblique photogrammetry (POP) approach to achieve georeferenced panoramic images and real 3D models together by using panorama image projection algorithms \cite{Zhang}.
\begin{figure}[htbp!]
\centering
\includegraphics[width=0.48\textwidth]{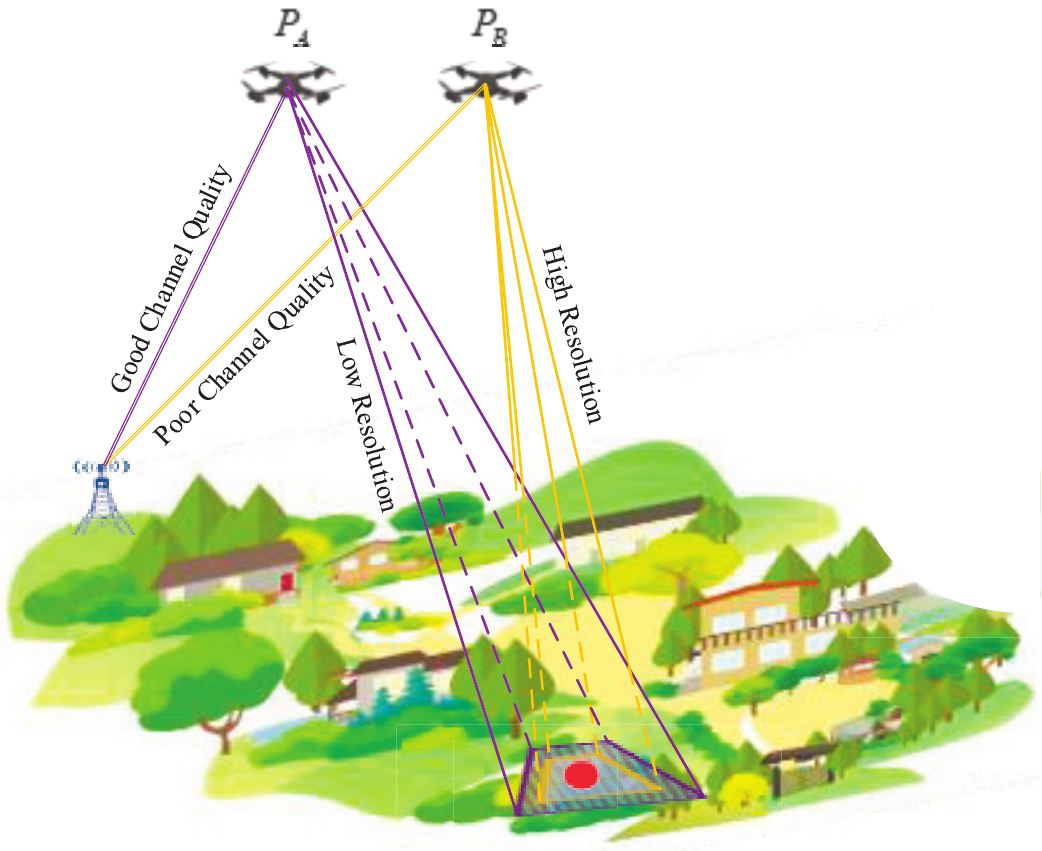}
\caption{The tradeoff between the image resolution and channel quality in the UAV-assisted image acquisition system.}
\label{F1}
\end{figure}

Although the above-mentioned literatures have made a good progress in aerial survey and object measurement, the oblique photography model they adopted is not suitable for UAV-assisted image acquisition. On one hand, the oblique photography model they adopted still can't make sure that the GT is imaged at the center of the photograph all the time which is not convenient for users to watch. On the other hand, the resolution metric is only slightly modified according to the traditional vertical photography model (i.e., adding the effect of the oblique angle) \cite{Mavrinac}, which lacks a more comprehensive formula related to the UAV's 3D coordinates. To address these issues, we propose a novel oblique photography model in this paper which quantitatively analyzes the influence of UAV's space position on the image quality. Based on the proposed oblique photography model, we further consider a UAV-assisted image acquisition and transmission system where the UAV is adopted to capture images of a GT via the carried camera and then transmit the captured image to the BS through the wireless backhaul. Fig. \ref{F1} shows a common scenario which is often encountered in practice: 1) at point $P_A$, the captured image has a low resolution, but the channel quality is good, and 2) at point $P_B$, the captured image has a high resolution, but the channel quality is poor. How to choose between these two points may make people confused. Therefore, the goal of this paper is to study the UAV deployment problem where the transmission delay can be minimized while satisfying the predetermined resolution requirement. \\
\textbf{Contributions}: The main contributions of this paper are three-fold:
\begin{enumerate}
\item A UAV-assisted oblique photography model is created where the resolution of the captured image is determined by the 3D coordinate of the UAV as well as the scale of the GT, which thus provides several new potential research directions. 
\item The UAV deployment problem is modeled as a non-convex optimization problem, aiming at minimizing the data transmission delay while ensuring that the predetermined resolution requirement can be satisfied. We firstly simplify the three-variable original problem into a two-variable one and then propose a suboptimal solution to the simplified problem by using the block coordinate descent (BCD) and successive convex approximation (SCA) techniques. 
\item Detailed numerical results are provided to verify the effectiveness of the proposed system. Firstly, the effects of horizontal and vertical coordinates on the resolution of UAV are analyzed. In addition, the performance of the UAV deployment problem are compared under three different solutions including 1) conventional scheme, 2) the proposed scheme solved by ES, and 3) the proposed scheme solved by BCD and SCA.  
\end{enumerate}

The reminder of the paper is organized as follows. Section II describes models of UAV-BS channel, oblique photography, and image transmission. In Section III, the original optimization problem is stated and simplified. In Section IV, an effective BCD and SCA-based algorithm is proposed to solve the simplified non-convex problem. In Section V, the numerical results are provided to show the effectiveness of the proposed system. In Section VI, we conclude this paper. In Section VII, we provide some potential research directions.

\vspace{0.2in}
\section{System Model}
Consider an UAV-assisted image acquisition system where a rotary-wing UAV is deployed to capture an image for a ground target (GT) and transfer the captured image data back to the base station (BS) immediately. In the following subsections, models of UAV-BS channel, oblique photography, and image transmission are described, respectively.

\subsection{UAV-BS Channel Model}
In this paper, we consider a real-time image acquisition system with a UAV, a BS, and a GT, where a three-dimensional (3D) Cartesian coordinate system is adopted. We assume that the coordinates of the BS and the GT are known a priori to the UAV. Let $\left( {{\bf{w}}_b^T,{z_b}} \right)$ denote the 3D coordinate of the BS, where ${{\bf{w}}_b}\!=\!{\left[ {{x_b},{y_b}} \right]^T}\!\in\!{\mathbb{R}^{2 \times 1}}$ represents the horizontal coordinate and ${z_b}$ represents the vertical coordinate, respectively. Similarly, let $\left( {{\bf{q}}^T,{z}} \right)$ denote the 3D coordinate of the UAV, where ${{\bf{q}}}\!=\!{\left[ {{x},{y}} \right]^T}\!\in\!{\mathbb{R}^{{\rm{2}} \times 1}}$ represents the horizontal coordinate and ${z}$ represents the vertical coordinate, respectively. As such, the distance between the UAV and the BS, denoted by $d_{u,b}$, is given by 
\begin{equation}\label{E1}
{d_{u,b}} = \sqrt {{{\left\| {{{\bf{q}}} - {{\bf{w}}_b}} \right\|}^{\rm{2}}}{\rm{ + (}}{z} - {z_b}{)^2}}.
\end{equation}

For ease of analysis, we assume that the wireless channel between the UAV and BS is dominated by the LoS link.\footnote{Note that this work can be extended to more complex channel models such as probabilistic LoS model or Rician fading model \cite{You2}. However, we may start with the simplest case under the LoS channel, where a trade-off between the image resolution and transmission rate still exists due to the possibly large horizontal distances between the GTs and the BS.} Hence, the channel power gain between the UAV and BS, denoted by $h_{u,b}$, can be modeled based on the free-space path loss model by 
\begin{equation}\label{E2}
{h_{u,b}} = {\beta _0}d_{u,b}^{-2} = \frac{\beta_0}{{{{\left\| {{{\bf{q}}} - {{\bf{w}}_b}} \right\|}^{\rm{2}}}{\rm{ + (}}{z} - {z_b}{)^2}}},
\end{equation}
where ${\beta_0}$ is the average channel power gain at the reference distance of $1$ meter. The achievable rate in bits per second (bps) between the UAV and BS, denoted by $R$, can thus be expressed as 
\begin{align}\label{E3}
{\mathcal{R}} = &\mathcal{B}{\log _2}\left( {1 + \frac{{|{h_{u,b}}{|^2}P}}{{{\sigma ^2}\Gamma }}} \right) \notag \\
= &\mathcal{B}{\log _2}\left( {1 + \frac{\gamma_0}{{({{\left\| {{{\bf{q}}} - {{\bf{w}}_b}} \right\|}^{\rm{2}}}{\rm{ + (}}{z} - {z_b}{)^2})}}} \right),
\end{align}
where $\mathcal{B}$ denotes the channel bandwidth, $\sigma ^2$ denotes the noise power, $P$ denotes the transmit power, and $\Gamma$ denotes the signal-to-noise ratio (SNR) gap between the practical modulation-and-coding scheme and the theoretical Gaussian signaling. For simplicity, we define $\gamma_0 = \frac{P\beta_0}{\sigma^2\Gamma}$ as the received SNR at the reference distance of 1 meter.

\subsection{Oblique Photography Model}
In conventional UAV-assisted vertical photography model, capturing images over the GT can achieve high resolution due to the very close distance to the GT. However, the channel quality between the BS and UAV is poor especially when the GT is far away from the BS, thus leading to a large transmission delay. As such, we present in Fig. \ref{F2} a UAV-assisted 3D oblique photography model where the UAV can capture images of the GT at a certain oblique angle from a distance rather than over it. Compared to the conventional vertical photography model, the oblique photography model can greatly reduce the transmission delay while satisfying the pre-determined resolution requirement.
\begin{figure}[htbp!]
\centering
\includegraphics[width=0.4\textwidth]{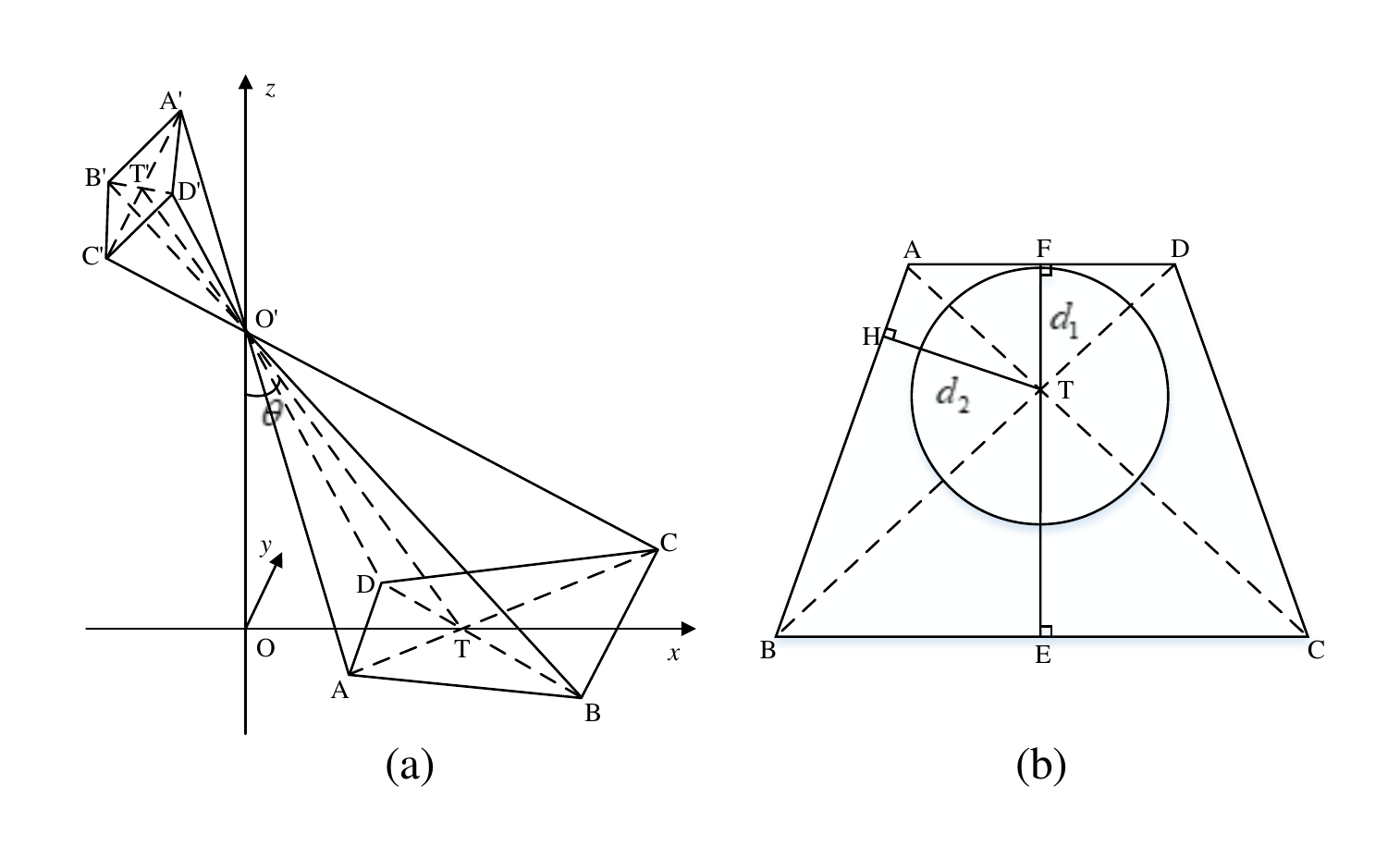}
\caption{(a) UAV-assisted 3D oblique photography model; (b) the location relationship between the GT and camera's coverage on the ground.}
\label{F2}
\end{figure}

In Fig. \ref{F2}(a), the rectangle $\rm{A'B'C'D'}$ represents the image plane of the camera and the isosceles trapezoid $\rm{ABCD}$ represents the camera's coverage region on the ground. $\rm{T'}$ and $\rm{T}$ represent the center of the image plane and GT, respectively. $\rm{O'}$ represents the focal point of the camera and $\rm{O}$ represents the projection of $\rm{O'}$ on the ground. We assume that the camera equipped on the UAV can automatically focus via adjusting its oblique angle to ensure that $\rm{T}$ is imaged at $\rm{T'}$. Let $\left( {{\bf{w}}_g^T,{z_g}} \right)$ denote the coordinate of GT, where ${{\bf{w}}_g}\!=\!{\left[ {{x_g},{y_g}} \right]^T}\!\in\!{\mathbb{R}^{2 \times 1}}$ represents the horizontal coordinate and ${z_g}$ represents the vertical coordinate. For simplicity, we assume that the GT is a circle with a known radius $r_0$ regardless of the influence of terrain fluctuations. Therefore, the height of GT is negligible, i.e., ${z_g}=0$. Denote by ${d_{u,g}}$ the distance from the UAV to GT, i.e., $|\rm{T'T}|$, which is given by
\begin{equation}\label{E4}
{d_{u, g}} = \sqrt {{{\left\| {{{\textbf{q}}} - {{\textbf{w}}_g}} \right\|}^{\rm{2}}}{\rm{ + }}{z}^2}.
\end{equation}

In conventional vertical photography model, the camera's image plane is parallel to the ground and the image resolution can be characterized by the ground sample distance (GSD) that each pixel can represent. However, in oblique photography model, the GSD that each pixel can represent is various. Hence, the image resolution can no longer be represented by GSD in oblique photography model. As such, we redefine the image resolution as the ratio of GT's area (a circle) to camera's coverage area (an isosceles trapezoid). Denote by $S_c$ the camera's coverage area when capturing the image of GT, is modeled as derived in Appendix  \ref{A}
\begin{equation}\label{E5}
\begin{aligned}
S_c = S_v \cdot \phi \left( {{\theta}} \right) = \frac{{{w_0}{l_0}z^2}}{{f_0^2}} \cdot \frac{1}{{{{\left( {1 - \frac{{w_0^2}}{{4f_0^2}}{{\tan }^2}{\theta}} \right)}^2}{{\cos }^3}{\theta}}},
\end{aligned}
\end{equation}
where ${\theta}$ denotes the camera's oblique angle (i.e., $\angle \rm{OO'T}$ in Fig.  \ref{F2}(a)) and we have $\cos {\theta}{\rm{ = }}\frac{{{z}}}{{{d_{u, g}}}}$ and $\tan {\theta} = \frac{{\left\| {{{\textbf{q}}} - {{\textbf{w}}_g}} \right\|}}{{{z}}}$. To guarantee that the UAV can successfully capture the image, ${\theta}$ should satisfy the constraint that $0\le{\theta}<\theta_0$,\footnote{$\theta_0 = \arctan\frac{{2f_0}}{w_0}$ represents the 
angle limit of oblique photography model. Please see Appendix A for more details.}   which is equivalent to
\begin{equation}\label{E6}
b_1{z} -\left\| {{{\textbf{q}}} - {{\bf{w}}_g}} \right\| \geq 0,
\end{equation}
where $f_0$ is the camera's focal length, and $w_0$ and $l_0$ represent the width and length of the image plane, respectively. To be specific, $S_v$ is equivalent to camera's coverage area when capturing images over the GT, which is only determined by $z$. $\phi \left( {{\theta}} \right)$ represents the coverage scaling factor which is monotonically increasing with respect to ${\theta}$. It can be indicated from (\ref{E5}) that when ${\theta}$ tends to $0$, the oblique photography model reduces to the conventional vertical photography model.

Accordingly, the redefined image resolution can be denoted by the UAV's 3D coordinate as follows.
\begin{equation}\label{E7}
{{\cal I}}{\rm{ = }}\frac{{{a{\left( {z^2 - \frac{{1}}{{b_1^2}}{{\left\| {{{\textbf{q}}} - {{\textbf{w}}_g}} \right\|}^{\rm{2}}}} \right)}^2}}}{{{{\left( {{{\left\| {{{\textbf{q}}} - {{\textbf{w}}_g}} \right\|}^{\rm{2}}}{\rm{ + }}z^2} \right)}^{\frac{{\rm{3}}}{{\rm{2}}}}}z^{\rm{3}}}},
\end{equation}
where $a = \frac{{b_1b_2\pi r_0^2}}{{4}}$ with $b_1=\frac{2f_0}{w_0}$ and $b_2=\frac{2f_0}{l_0}$ being constants related to the camera setting. 


As illustrated in Fig. \ref{F2}(b), the GT should be completely included in the camera's coverage region, thus leading to the following constraint
\begin{equation}\label{E8}
{r_0} \le \min \left( {d_1, d_2} \right),
\end{equation}
where $ d_1$ and $d_2$ represent the distance from point $\rm{T}$ to $\rm{AD}$ (i.e., $|\rm{FT}|$) and $\rm{BC}$ (i.e., $|\rm{HT}|$), respectively, which are given in the form of (see Appendix \ref{B} for details)
\begin{equation}\label{E9}
d_1{\rm{ = }}\frac{{z^2 + {{\left\| {{{\textbf{q}}} - {{\textbf{w}}_g}} \right\|}^{\rm{2}}}}}{{b_1{z} + \left\| {{{\textbf{q}}} - {{\textbf{w}}_g}} \right\|}}
\end{equation}
\begin{equation}\label{E10}
d_2{\rm{ = }}\frac{{z^2 + {{\left\| {{{\textbf{q}}} - {{\textbf{w}}_g}} \right\|}^{\rm{2}}}}}{{{{\left( {b_2^2z^2 + (1 + b_2^2){{\left\| {{{\textbf{q}}} - {{\textbf{w}}_g}} \right\|}^2}} \right)}^{\frac{1}{2}}}}}
\end{equation}

\subsection{Image Transmission Model}
Let $\bar {\mathcal{F}}$ denote the size of the uncompressed captured image. Specifically, $\bar {\mathcal{F}} = \frac{{w_0l_02^n}}{{{\delta_0^2}}}$ where $\delta_0$ represents the length of each pixel and $n$ represents the bit depth of the image.{\footnote{The bit depth means the number of bits used to hold a pixel.}} Thus, $\bar {\mathcal{F}}$ is a constant related to the camera settings. In this paper, we assume that only the image data containing the GT will be transmitted back to the BS.{\footnote{Image data containing the GTs can be easily differentiated via some specific separation models such as low-rank decomposition model, Gaussian mixture model as well as deep learning based model \cite{Tang2}.}}  Therefore, the transmission time of the captured image, denoted by $\mathcal{T}$, can be represented as follows 
\begin{equation}\label{E11}
\mathcal{T}  = \frac{\alpha{\bar {\mathcal{F}}\mathcal{Q}}}{{\mathcal{R}}},
\end{equation}
where $0 \le \alpha \le 1$ denotes the compression ratio,{\footnote{Specific image encoders, e.g., JPEG, can remove redundancy by leveraging the inter-pixel correlations \cite{He1}. Moreover, the smaller $\alpha$ is, the less data needs to be transmitted. However, too small compression ratio will lead to the degradation of image quality. JPEG is able to achieve a compression ratio up to 0.1 without visible loss in image quality.}} and ${\alpha{\bar {\mathcal{F}}\mathcal{Q}}}$ denotes the amount of the transmitted image data in bits. From (\ref{E11}), one can see that $\mathcal{T}$ is determined by both $\mathcal{Q}$ and $\mathcal{R}$. Specifically, both large $\mathcal{Q}$ and small $\mathcal{R}$ will result in a large transmission delay. 

\section{Optimization Problem}
In this section, we firstly state the original optimization problem which is a non-convex problem. Then, we further reformulate the original problem into a more tractable form. 

\subsection{Problem Statement}
In this paper, we aim to find an optimal shooting point for the UAV in order to minimize the data transmission time while ensuring that the UAV can capture the image successfully as well as satisfying the predetermined resolution requirement. As such, the optimization problem can be formulated as follows
\begin{subequations}
\begin{align}
(\textbf{P1})~~\mathop {\min }\limits_{ {{\bf{q}}},{z} }\;\;&\mathcal{T} \notag \\
\rm{s.t.}~~& \mathcal{I} \ge {\mathcal{I}_{\min }}, \label{E12a} \\
&b_1{z} - \left\| {{{\bf{q}}} - {{\bf{w}}_g}} \right\| \geq 0, \label{E12b}\\
&r_0\le \min(d_1, d_2), \label{E12c}
\end{align}
\end{subequations}
where ${\mathcal{I_{\min}}}$ represents the minimum resolution requirement. (\ref{E12a}) indicates that the resolution of the captured image should satisfy the minimum requirement. (\ref{E12b}) shows the condition that the UAV can capture images effectively. (\ref{E12c}) guarantees that the GT can be completely included in the camera's coverage region. Since $\{ {{\bf{q}}},{z}\}$ is a set of coupling variables and (\ref{E12a}) and (\ref{E12c}) as well as the objective function are non-convex, (P1) is an intractable non-convex optimization problem.

\subsection{Problem Reformulation}
In this part, we will propose a new objective function equivalent to that in (P1). To start with, we will first analyze the influence of UAV's horizontal coordinate on the transmission time with fixed UAV's vertical coordinate. Before analyzing this problem, we introduce the following two properties.
\newtheorem{property}{\emph{\underline{Property}}}
\begin{property}
\label{p1}
$f(x) = \frac{{{m_1}{{\left( {{m_2} - {x^2}} \right)}^2}}}{{{{\left( {{x^2}{\rm{ + }}{m_0}} \right)}^{\frac{{\rm{3}}}{{\rm{2}}}}}}},{m_0} > 0,{m_1} > 0,{m_2} > 0$ is monotonically decreasing in the feasible domain of $0 \le x < \sqrt {{m_2}}$.
\end{property}
\begin{proof}
See \emph{Appendix C}.
\end{proof}

\begin{property}
\label{p2}
$\varphi (x) = {\log _2}\left( {1 + \frac{{{n_1}}}{{{x^2} + {n_0}}}} \right),{n_0} > 0,{n_1} > 0$ is monotonically decreasing in the feasible domain of $x\ge 0$.
\end{property}

According to Property \ref{p1}, with given $z$, we can conclude that $\mathcal{Q}$ decreases with the increase of $\left\| {{{\bf{q}}}\!-\!{{\bf{w}}_g}} \right\|$ in the feasible domain of $0 \le \left\| {{{\bf{q}}}\!-\!{{\bf{w}}_g}} \right\| \le b_1z$.
We can also conclude from Property \ref{p2} that $\mathcal{R}$ decreases with the increase of $\left\| {{{\bf{q}}}\!-\!{{\bf{w}}_b}} \right\|$. Combining Property \ref{p1} with Property \ref{p2}, we have the following lemma stand.
\newtheorem{lemma}{\emph{\underline{Lemma}}}
\begin{lemma}
\label{l1}
The optimal solution to problem (P1) must satisfy ${{\bf{q}}} = \eta {{\bf{w}}_b} + (1-\eta){{\bf{w}}_g}$, where $\eta  \in {[0,1]}$ represents the horizontal coordinate indicator.
\end{lemma}
\begin{proof}
See \emph{Appendix D}.
\end{proof}

According to Lemma \ref{l1}, we can conclude that the optimal shooting point should be located at the plane through the GT and BS, where the image resolution $\mathcal{Q}$ decreases with the increase of the transmission rate $\mathcal{R}$ with fixed $z$. Therefore, the minimum transmission time $\mathcal{T}$ can be obtained when the transmission rate is maximum. Therefore, (P1) can be reformulated as follows
\begin{align}
(\textbf{P2})\mathop{\max}\limits_{ {{\bf{q}}},{z}}~~&\mathcal{R} \notag \\
\rm{s.t.}~~~&(\ref{E12a})-(\ref{E12c}). \notag
\end{align}

According to Lemma \ref{l1}, the horizontal distance between the UAV and the BS (i.e., $\left\|{{{\bf{q}}}\!-\!{{\bf{w}}_b}}\right\|$) as well as the GT (i.e., $\left\|{{{\bf{q}}}\!-\!{{\bf{w}}_g}}\right\|$) can be represented as follows
\begin{align}
\left\|{{{\bf{q}}}\!-\!{{\bf{w}}_b}}\right\| = (1-\eta)d_{g, b}, \text{and} \left\|{{{\bf{q}}}\!-\!{{\bf{w}}_g}}\right\| =  {\eta}d_{g, b}, \label{E14}
\end{align}
where $d_{g, b} = \left\|{{{\bf{w}}_g}\!-\!{{\bf{w}}_b}}\right\|$ denotes the horizontal distance between the BS and the GT which is known a priori. As such, ${\mathcal{R}}$ can be re-represented as 
\begin{align}\label{E15}
\mathcal{R} = {\log _2}\left( {1 + \frac{\gamma_0}{{(1-\eta)^2d_{g, b}^2{\rm{ + (}}{z} - {z_b}{)^2} }}} \right). 
\end{align}
Therefore, maximizing $\mathcal{R}$ in (P2) is equivalent to minimizing $(1-\eta)^2d_{g, b}^2{\rm{ + (}}{z} - {z_b}{)^2}$. 

In the following, we will convert (\ref{E12a})-(\ref{E12c}) into more tractable forms by leveraging (\ref{E15}). Firstly, by taking logarithm for both sides, (\ref{E12a}) can be re-represented as follows 
\begin{align}\label{E16}
2\ln{{( {z^2 - {\eta}^2\frac{1}{b_1^2}d_{g,b}^2})}}-\frac{3}{2}\!\ln{( {z^2\!+\!{\eta}^2d_{g,b}^2})}-\!3\ln{z} \ge \ln\frac{\mathcal{I}_{\min }}{a}.
\end{align}

Secondly, (\ref{E12b}) can be simplified into the following form 
\begin{align}\label{E17}
&b_1{z} - {\eta}d_{g, b}\geq 0.
\end{align}

Finally, (\ref{E12c}) can be rewritten as
\begin{align}\label{E18}
{{ {z^2\!+\!{\eta}^2d_{g,b}^2}}}\geq r_0\!\max\!\left({{{b_1z\!+\! \eta d_{g,b}}}},{{({b_2^2z^2\!+\!(1\!+\!b_2^2){\eta}^2d_{g,b}^2})^{\frac{1}{2}}}}\right), 
\end{align}
of which the right-hand side (RHS) is convex now due to the convexity of the maximization function.

As such, (P2) can be reformulated as follows.
\begin{subequations}
\begin{align}
(\textbf{P3})\mathop{\min}\limits_{ {\eta},{z}}~~&(1-\eta)^2d_{g,b}^2{\rm{ + (}}{z}\!-\!{z_b}{)^2} \notag \\
\rm{s.t.}~~~& 0\le \eta \le 1,\label{E19a}\\
&(\ref{E16})-(\ref{E18}), \notag
\end{align}
\end{subequations}
where the objective function is convex now. Comparing to the original optimization (P1), the variables reduce from three dimensions to two dimensions. However, (\ref{E16}) and (\ref{E18}) are still non-convex. In the following section, BCD and SCA techniques will be applied to obtain the sub-optimal solution to (P3).

\vspace{0.2in}
\section{Proposed Algorithm}
In this section, we will propose an efficient iterative algorithm to obtain a high-quality suboptimal solution to (P3) by applying BCD and SCA techniques \cite{Liu}. Specifically, (P3) is tackled by iteratively solving two subproblems to optimize the horizontal coordinate indicator and the vertical coordinate with the other one being fixed, until the algorithm converges to a given threshold. We also provide detailed analysis on the complexity of the proposed algorithm.
\subsection{Horizontal Coordinate Indicator Optimization with Fixed Vertical Coordinate}
Given the UAV's vertical coordinate (i.e., $z$), we first consider the following sub-optimization problem of (P3) to optimize the UAV's horizontal coordinate indicator (i.e., $\eta$). 
\begin{align}
(\textbf{P3.1})\mathop{\max}\limits_{ {\eta}}~~&\eta \notag \\
\rm{s.t.}~~&(\ref{E16})-(\ref{E18}), (\ref{E19a}). \notag
\end{align}

Problem (P3.1) is still a non-convex problem due to the non-convexity of the constraints (\ref{E16}) and (\ref{E18}). Specifically, given $z$, the third term (i.e., $-\frac{3}{2}\!\ln{( {z^2\!+\!{\eta}^2d_{g,b}^2})}$) in the left-hand side (LHS) of (\ref{E16}) is convex with respect to (w.r.t) $\eta^2$. Then, we can derive its convex approximation by applying SCA technique as follows
\begin{align}\label{E20}
-\frac{3}{2}\!\ln{( {z^2\!+\!{\eta}^2d_{g,b}^2})} \ge -\frac{3}{2}\!\ln{\left( {z^2\!+\!{\widehat \eta}^2d_{g,b}^2}\right)} \notag \\
-\frac{3d_{g,b}^2}{2{\left( {{z}^2\!+\!{\widehat \eta}^2d_{g,b}^2}\right)}}\left({\eta}^2-{\widehat \eta}^2\right),
\end{align}
where ${\widehat \eta}$ is a local point of $\eta$. According to (\ref{E20}), the constraint (\ref{E16}) can be rewritten as
\begin{align}\label{E21}
&2\ln{{( {z^2 - \frac{1}{b_1^2}{\eta}^2d_{g,b}^2})}}-\frac{3}{2}\!\ln{\left( {z^2\!+\!{\widehat \eta}^2d_{g,b}^2}\right)} \notag \\
&-\frac{3d_{g,b}^2}{2{\left( {{z}^2\!+\!{\widehat \eta}^2d_{g,b}^2}\right)}}\left({\eta}^2-{\widehat \eta}^2\right)-\!3\ln{z} \ge \ln\frac{\mathcal{I}_{\min }}{a},
\end{align}
which is a convex constraint with given $z$ now.

To solve the non-convexity of (\ref{E18}), we can derive the convex approximation of the second term of the LHS by leverage the SCA technique as follows. 
\begin{align}\label{E22}
{\eta}^2d_{g,b}^2 \ge {\widehat \eta}^2d_{g,b}^2 + 2{\widehat \eta}({\eta}-{\widehat \eta})d_{g,b}^2.
\end{align}

As a result, with given $z$, the constraint (\ref{E18}) can be rewritten into a convex one as follows
\begin{align}\label{E23}
&{{ {z^2\!+\!{\widehat \eta}^2d_{g,b}^2 + 2{\widehat \eta}({\eta}-{\widehat \eta})d_{g,b}^2}}}\ge \notag \\
&r_0\!\max\!\left({{{b_1z\!+\! \eta d_{g,b}}}},{{({b_2^2z^2\!+\!(1\!+\!b_2^2){\eta}^2d_{g,b}^2})^{\frac{1}{2}}}}\right), 
\end{align} 

Then, problem (P3.1) can be reformulated as follows
\begin{align} 
(\textbf{P3.1.1})\mathop{\max}\limits_{ {\eta}}~&\eta \notag \\
\rm{s.t.}~~&(\ref{E17}), (\ref{E19a}), (\ref{E21}), (\ref{E23}). \notag
\end{align}

The constraints (\ref{E21}) and (\ref{E23}) are convex with respect to $\eta$ now. Consequently, the sub-optimization problem (P3.1.1) is a standard convex optimization problem which can be solved through some optimization toolboxes, e.g., CVX.

\subsection{Vertical Coordinate Optimization with Fixed Horizontal Coordinate Indicator}
Next, given the UAV's horizontal location (i.e., $\eta$), we consider the following sub-optimization problem of (P3) to optimize the UAV's vertical location (i.e., $z$).
\begin{align}
(\textbf{P3.2})\mathop{\min}\limits_{{z}}~&{\rm{ (}}{z}\!-\!{z_b}{)^2} \notag \\
\rm{s.t.}~~&(\ref{E16})-(\ref{E18}). \notag
\end{align}

Problem (P3.2) is also a non-convex problem due to the non-convexity of the constraints (\ref{E16}) and (\ref{E18}). Specifically, given $\alpha$, in the LHS of (\ref{E16}), the third term (i.e., $-\frac{3}{2}\!\ln{( {z^2\!+\!{\eta}^2d_{g,b}^2})}$) is convex with respect to $z^2$ and the fourth term (i.e., $-3\ln z$) is convex with respect to $z$. Similarly, we can derive their convex approximations by using SCA technique as follows
\begin{align}\label{E24}
-\frac{3}{2}\!\ln{( {z^2\!+\!{\eta}^2d_{g,b}^2})} \ge -\frac{3}{2}\!\ln{\left( {{\widehat z}^2\!+\!{\eta}^2d_{g,b}^2}\right)} \notag \\
-\frac{3}{2{\left( {{{\widehat z}}^2\!+\!{\eta}^2d_{g,b}^2}\right)}}\left({z}^2-{\widehat z}^2\right),
\end{align}

\begin{align}\label{E25}
-3\ln{z} \ge -3\ln{\widehat z} -\frac{3}{{\widehat z}}(z-{\widehat z}),
\end{align}
where $\widehat z$ is a local point of $z$. Then, the constraint (\ref{E16}) can be rewritten as the following convex form
\begin{align}\label{E26}
&\!2\ln{{( {z^2 - \frac{1}{b_1^2}{\eta}^2d_{g,b}^2})}}-\frac{3}{2}\!\ln{\left( {{\widehat z}^2\!+\!{\eta}^2d_{g,b}^2}\right)}-  \notag \\
&\frac{3}{2{({{{\widehat z}}^2\!+\!{\eta}^2d_{g,b}^2})}}\left({z}^2\!-\!{\widehat z}^2\right)\!-\!3\ln{\widehat z}\!-\!\frac{3}{{\widehat z}}(z\!-\!{\widehat z})\!\ge\!\ln\ln\frac{\mathcal{I}_{\min }}{a}. 
\end{align}

With given $\eta$, the convex approximation of the first term in (\ref{E18}) can be derived by applying the SCA technique as 
\begin{align}\label{E27}
z^2 \ge {\widehat z}^2 + 2{\widehat z}(z- {\widehat z}).
\end{align}

Accordingly, the constraint (\ref{E18}) can be rearranged into the new convex form as follows
\begin{align}\label{E28}
&{{ {{\widehat z}^2 + 2{\widehat z}(z- {\widehat z})\!+\!{\eta}^2d_{g,b}^2}}}\ge \notag \\
&r_0\!\max\!\left({{{b_1z\!+\! \eta d_{g,b}}}},{{({b_2^2z^2\!+\!(1\!+\!b_2^2){\eta}^2d_{g,b}^2})^{\frac{1}{2}}}}\right).
\end{align}

As such, problem (P3.2) can be reformulated as the following sub-optimization problem (P3.2.1)
\begin{align}
(\textbf{P3.2.1})\mathop{\min}\limits_{{z}}~~&{\rm{ (}}{z}\!-\!{z_b}{)^2} \notag \\
\rm{s.t.}~~~&(\ref{E18}, (\ref{E26}), (\ref{E28}). \notag
\end{align}

The constraints (\ref{E26}) and (\ref{E28}) are convex with respect to $z$ now. Therefore, the sub-optimization problem (P3.2.1) is a standard convex optimization problem which can be solved through some optimization toolboxes, e.g., CVX.

\subsection{Complexity and Convergence Analysis}
Based on the results of the above two sub-optimization problems, the overall algorithm for computing the sub-optimal solution to (P3) is summarized in \emph{Algorithm 1}. The complexity of Algorithm 1 is analyzed as follows. In each iteration, the horizontal coordinate indicator (i.e., $\eta$) and the vertical coordinate (i.e., $z$) are iteratively optimized using the convex solver based on the interior-point method, and thus their individual complexity can be represented as $O(\rm{log}(1/\varsigma))$ and $O(\rm{log}(1/\varsigma))$, respectively. Specifically, $\varsigma$ represents the predetermined solution accuracy. Then accounting for the BCD iterations with the complexity in the order of $\rm{log}(1/\varsigma)$, the total computation complexity of Algorithm 1 is $O(\rm{log}^2(1/\varsigma))$.
\begin{algorithm}[htb]
\caption{Iterative optimization for {{$\eta$}} and {{$z$}}.}
\label{alg:Framework}
\begin{algorithmic}[1]
\State Initialize ${\eta^0}$ and ${{{z}}^0}$. Let iteration index $i = 0$
\State {\bf{repeat}}
\State Solve problem (P3.1.1) with given UAV's vertical location ${{{z}}^i}$, and denote the optimal solution to the UAV's horizontal location as ${{{\eta}}^{i+1}}$.
\State Solve problem (P3.2.1) with given UAV's horizontal location ${{{\eta}}^{i+1}}$, and denote the optimal solution to the UAV's vertical location as ${{{z}}^{i+1}}$.
\State Update $i = i+1$
\State {\bf{until}} the computed objective value of problem (P3) converges within a pre-specified precision $\varsigma  > 0$.
\end{algorithmic}
\end{algorithm}

Next, we investigate the convergence property of Algorithm 1. Let $\mu({{{\eta}}^i},{{{z}}^i})$ denote the objective value of problem (P3) in the $i$-th iteration. Therefore, the following inequality holds,
\begin{equation}\label{E29}
\begin{array}{l}
\mu({{{\eta}}^i},{{{z}}^i}) \mathop \le \limits^{(a)} \mu({{{\eta}}^{i+1}},{{{z}}^i})
\!\mathop \le\limits^{(b)}\mu({{{\eta}}^{i+1}},{{{z}}^{i+1}})\mathop  \le \limits^{(c)}\!\mu^*({{{\eta}}^{i+1}},{{{z}}^{i+1}}).
\end{array}
\end{equation}
where $\mu^*({{{\eta}}^{i+1}},{{{z}}^{i+1}})$ represents the optimal solution to (P3). The inequality (a) holds since Step 3 in Algorithm 1 can obtain the optimal solution to (P3.1.1). The inequality (b) holds as Step 4 in Algorithm 1 can obtain the optimal solution to (P3.2.1). Since the SCA technique is used to achieve the lower bounds of the constraints (\ref{E16}) and (\ref{E18}), the optimal solution to (P3) must be the lower bound of the original optimal solution. Consequently, the inequality (c) holds. Therefore, the optimal solutions to (P3.1.1) and (P3.2.1) are guaranteed to be non-decreasing according to (\ref{E29}) over the iterations. Thus Algorithm 1 can converge to a locally optimal solution.

\section{Numerical Results}
In this section, we will provide numerical results to verify the effectiveness of the proposed system. To be specifically, we compare the performance of the proposed system with two other systems including conventional vertical photography scheme and proposed scheme solved by exhaustive search (ES). In the simulations, we assume that the BS and GT are located at $(0, 0, 25)$\;m and $(150, 200, 0)$\;m, respectively. The transmit power is set to $P = 10$\;dBm and average noise power is set to $\sigma_0^2=-109$\;dBm. $\beta_0$ is assumed to be $-40$\;dB and the SNR gap $\Gamma$ is set to $10$\;dB. Specifically, we set the radius of the GT as $r_0 = 20$\;m. According to the camera settings in \cite{Martinus}, $f_0$, $w_0$, $l_0$ as well as $\delta_0$ is set to $0.035$\;m, $0.0156$\;m, $0.0235$\;m, and $3.9\times 10^{-6}$\;m, respectively. The image compression ratio $\alpha$ is set to $0.8$ for simplicity. 

Recall that in this paper, we mainly aim at minimizing the data transmission time while satisfying the predetermined resolution requirement. Therefore, we present in Figs. \ref{F5} and \ref{F6} the performance comparison results of our proposed scheme with other two heuristic schemes, named conventional vertical photography scheme (for brevity, we denote it by conventional scheme) and proposed scheme solved by ES. To be specific, we assume that the UAV should fly to the top of the GT to capture the image with the same resolution in the conventional scheme. For another heuristic scheme, we find the optimal shooting point for the UAV by ES with a step size of $1m$ which thus usually has an extremely higher complexity with $O(n^3)$ than our proposed BCD and SCA-based algorithm. 
\begin{figure}[htbp!]
\centering
\includegraphics[width=0.5\textwidth]{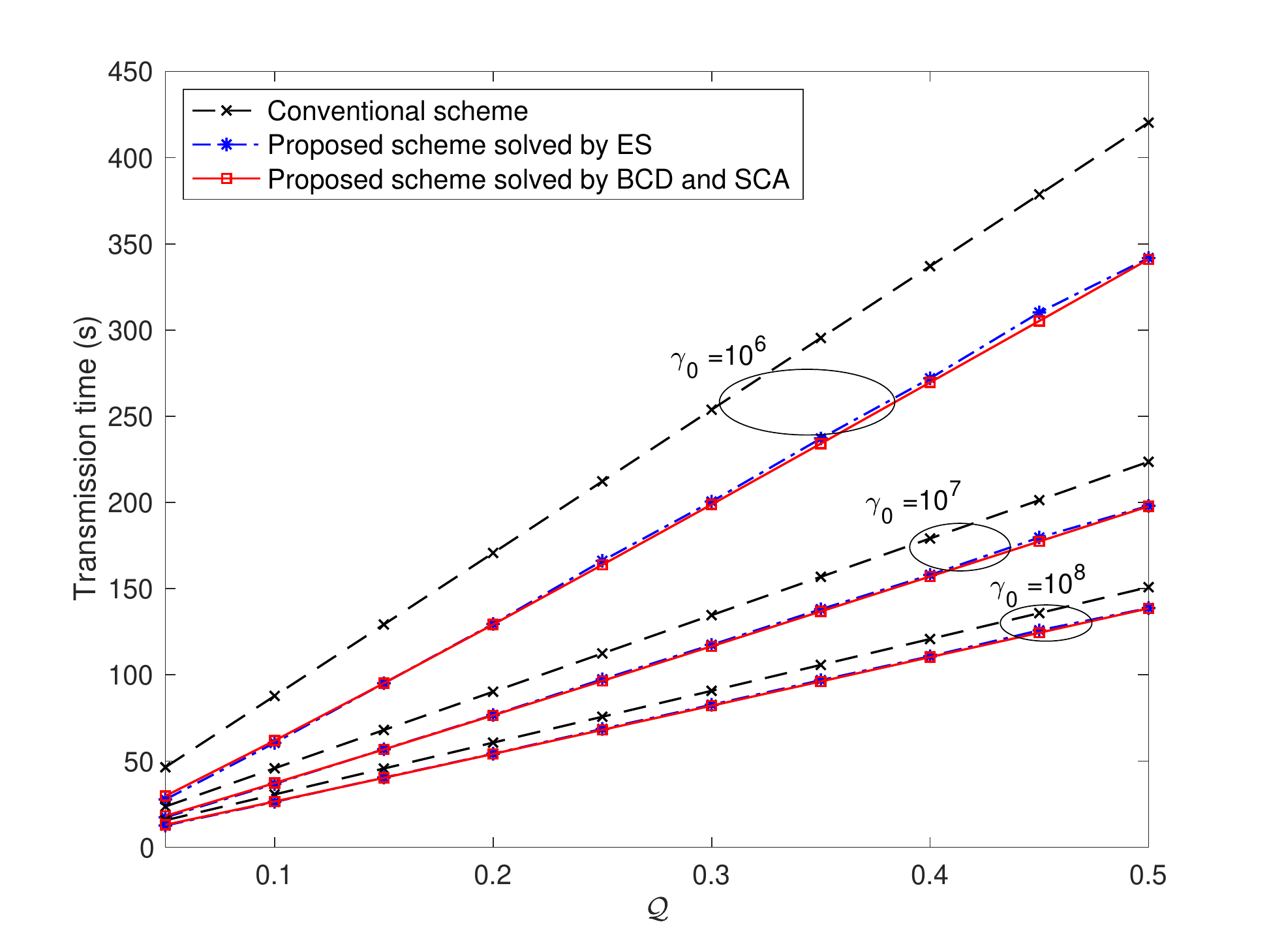}
\caption{Performance comparisons of transmission time versus $\mathcal{Q}$ among different schemes under different received SNR at the reference distance of 1 meter, i.e., $\gamma_0 = \{10^6, 10^7, 10^8\}$.}
\label{F5}
\end{figure}

In Fig. \ref{F5}, we show the performance comparison results of transmission time versus $\mathcal{Q}$ among different schemes under different received SNR at the reference distance of 1 meter, i.e., $\gamma_0 = \{10^6, 10^7, 10^8\}$. Please note that here we simulate different channel conditions via changing $\gamma_0$. From Fig. \ref{F5}, one can see that the proposed scheme solved by BCD and SCA can achieve a comparable performance in each case as compared to the proposed scheme solved by ES, thus verifying the correctness and effectiveness of our proposed model and optimization algorithm. One can also find from Fig. \ref{F5} that the proposed scheme can effectively save a lot of time as compared to the conventional scheme. Moreover, such superiority becomes obvious gradually with the degradation of the channel condition as well as with the increase of resolution requirement.   

\begin{figure}[htbp!]
\centering
\includegraphics[width=0.5\textwidth]{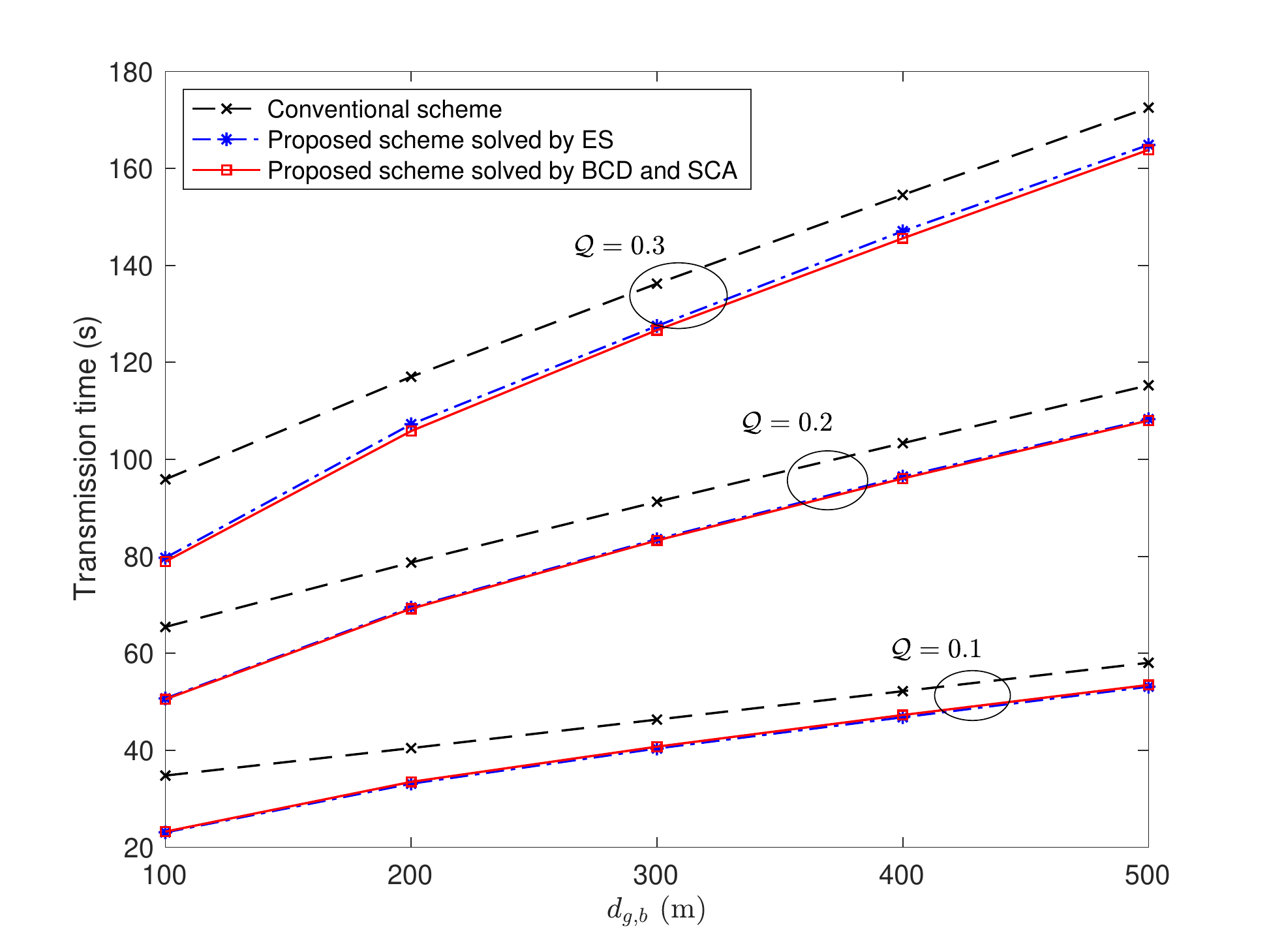}
\caption{Performance comparisons of transmission time versus $d_{g,b}$ among different schemes under different resolution requirements, i.e., $\mathcal{Q} = \{0.1, 0.2, 0.3\}$.}
\label{F6}
\end{figure}
In Fig. \ref{F6}, we present the performance comparison results of transmission time versus the distance between the BS and the GT, i.e., $d_{g,b}$ among different schemes under different resolution requirements, i.e., $\mathcal{Q} = \{0.1, 0.2, 0.3\}$. Again, the similar performance of the proposed scheme solved by BCD and SCA technique or by ES verify the correctness of the proposed algorithm. With the increase of the distance between the BS and the GT, our proposed scheme can save more transmission time as compared to the conventional scheme. In addition, as the resolution requirement becomes more stringent, our proposed scheme can grab more performance gains as compared to the conventional scheme. This is expected because with the increase of $\mathcal{Q}$, the data amount of the captured image becomes much larger which can be referred from (\ref{E11}). However, our proposed scheme can provide a non-trivial tradeoff between the resolution and the transmission delay, thus greatly reducing the transmission time. 

\section{Conclusion}
In this paper, we have proposed a novel oblique photography model, where the image resolution is redefined based on the UAV's 3D coordinate. By leveraging the proposed oblique photography model, the UAV placement problem has been formulated into a non-convex optimization problem, aiming at minimizing the data transmission time while satisfying the predetermined resolution requirement. The original problem is firstly simplified and then suboptimally solved with BCD and SCA techniques. Comprehensive numerical results have been presented to show the effectiveness of the proposed scheme.

\begin{appendices}
\section{}\label{A}
By solving the derivative of $f(x)$, we can obtain
\begin{equation}\label{E44}
f'(x)\!=\!\frac{{-{m_1}x({m_2}\!-\!{x^2}){{({x^2}\!+\!{m_0})}^{\frac{1}{2}}}({x^2}\!+\!4{m_0}\!+\! 3{m_2})}}{{{{({x^2}\!+\!{m_0})}^3}}}.
\end{equation}

Since $0\le x<\sqrt{m_2}$, $m_0>0$, and $m_1>0$, then we have $f'(x) \le 0$. Therefore, one can conclude that $f(x)$ decreases monotonically in the feasible domain of $0 \le x<\sqrt{m_2}$. 

\section{}\label{B}
\begin{figure}[htbp!]
\centering
\includegraphics[width=0.25\textwidth]{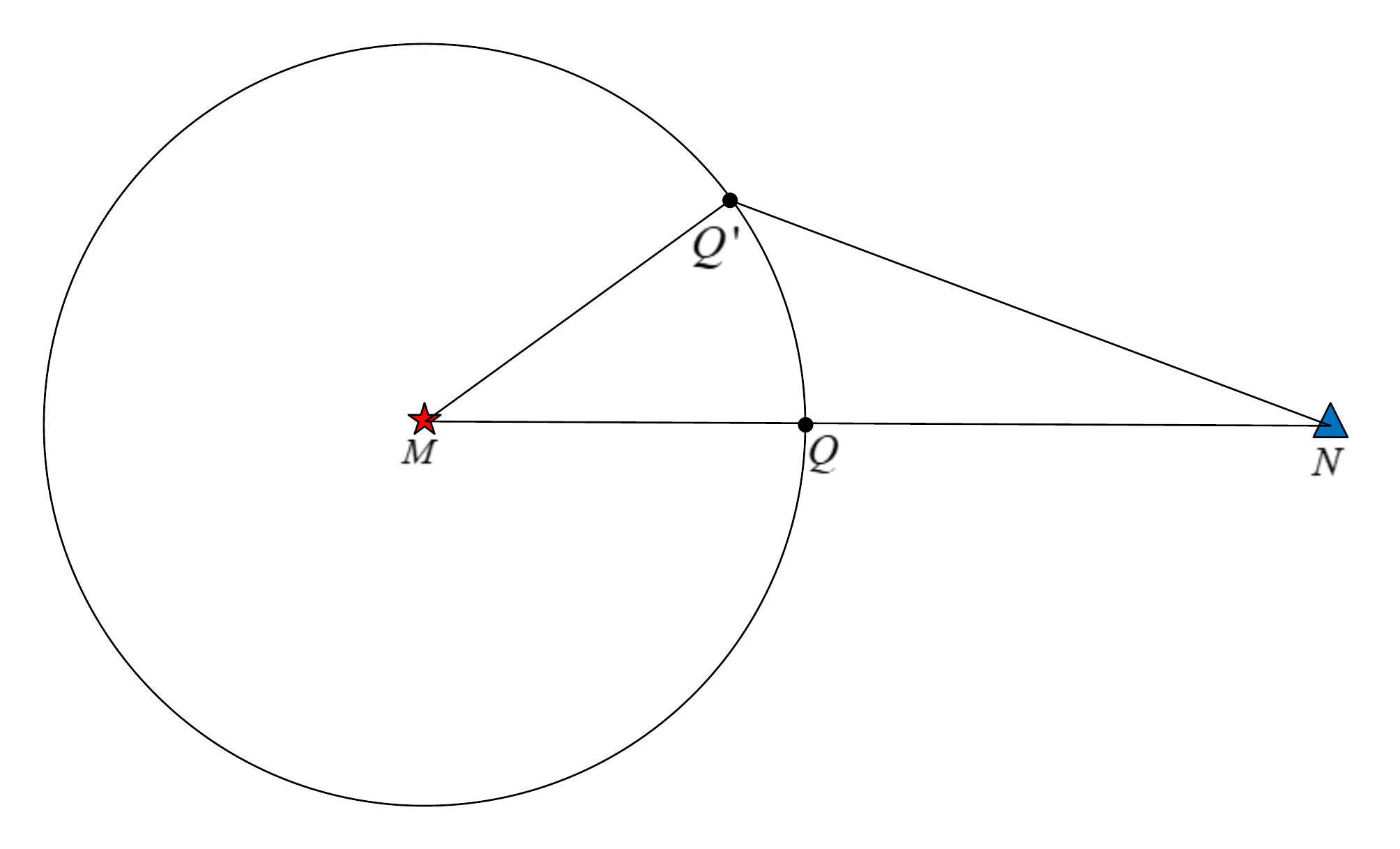}
\caption{Comparisons of transmission rate at different points with given $z$.}
\label{F8}
\end{figure}
For ease of exposition, we present in Fig. \ref{F8} a horizontal diagram to compare the transmission rate at different points. Assume that point ${\rm{M}}$ denotes the GT and point ${\rm{N}}$ denotes the BS, respectively. For given $z$, ${\rm{Q}}$ and ${\rm{Q'}}$ are two different points on the same circle centered at ${\rm{M}}$ and we have $\left\| {{\bf{q}}^{\rm{Q}} - {{\bf{w}}_g}} \right\|\!=\!\left\| {{\bf{q}}^{{\rm{Q'}}} - {{\bf{w}}_g}} \right\|$. To be specific, ${\rm{Q}}$ is located on the line segment of $\rm{MN}$ whose horizontal coordinate satisfies $\left\| {{{\bf{q}}^{\rm{Q}} } - {{\bf{w}}_b}} \right\| + \left\| {{{\bf{q}}^{\rm{Q}} } - {{\bf{w}}_g}} \right\| = \left\| {{{\bf{w}}_b} - {{\bf{w}}_g}} \right\|$. 

According to (\ref{E7}), we can conclude that images captured by the UAV have the same resolution at $\rm{Q}$ and $\rm{Q'}$ due to the same vertical altitude and the same distance from the GT. In $\triangle \rm{MQ'N}$, we have $\left\| {{\bf{q}}^{\rm{Q'}} - {{\bf{w}}_g}} \right\| + \left\| {{\bf{q}}^{\rm{Q'}} - {{\bf{w}}_b}} \right\| > \left\| {{{\bf{w}}_b} - {{\bf{w}}_g}} \right\| = \left\| {{{\bf{q}}^{\rm{Q}}} - {{\bf{w}}_b}} \right\| + \left\| {{{\bf{q}}^{\rm{Q}}} - {{\bf{w}}_g}} \right\|$. Since $\left\| {{\bf{q}}^{\rm{Q}} - {{\bf{w}}_g}} \right\| = \left\| {{\bf{q}}^{{\rm{Q'}}} - {{\bf{w}}_g}} \right\|$, then we have $\left\| {{\bf{q}}^{{\rm{Q'}}} - {{\bf{w}}_b}} \right\| > \left\| {{\bf{q}}^{\rm{Q}} - {{\bf{w}}_b}} \right\|$. Therefore, according to Property 2, we can draw the conclusion that the UAV can always achieve higher transmission rate at point $\rm{Q}$ than that at point $\rm{Q'}$ with given $z$. As such, the transmission time at point $\rm{Q}$ must be smaller than that at point  $\rm{Q'}$. Hence, the optimal solution to problem (P1) can always be found on the line segment of $\rm{MN}$. Equivalently, the horizontal coordinate of the optimal solution must satisfy ${{\bf{q}}} = \eta {{\bf{w}}_b} + (1-\eta){{\bf{w}}_g}$, where ${{\bf{w}}_b}$ represents the horizontal coordinate of $\rm{N}$, ${{\bf{w}}_g}$ represents the horizontal coordinate of $\rm{M}$, and $\eta  \in {[0,1]}$ represents the horizontal coordinate indicator.
\end{appendices}

\vfill

\bibliographystyle{IEEEtran}
\end{document}